\documentclass[12pt]{article}

\usepackage{amsmath,amssymb,amsfonts,mathrsfs,mathtools,dsfont,mathrsfs, amsthm, bm, bbm}
\DeclareMathAlphabet{\mathpzc}{OT1}{pzc}{m}{it}
\DeclareSymbolFontAlphabet{\amsmathbb}{AMSb}%
\newcommand{\GS}{\gamma^\star}
\usepackage{bm,bbm,dsfont,epsfig,rotating,setspace,latexsym,amsmath,epsf,amsthm,amssymb,amsfonts,epstopdf,graphicx}

\usepackage{cite,authblk}
\usepackage[dvipsnames]{xcolor}

\usepackage[normalem]{ulem}
\usepackage{comment}

\usepackage[dvipsnames]{xcolor}
\usepackage[font=small]{caption}
\usepackage[caption=false, font=small]{subfig}
\usepackage[hyphens]{url}
\usepackage[breaklinks, colorlinks, linkcolor=MidnightBlue, anchorcolor=MidnightBlue, citecolor=MidnightBlue, urlcolor=MidnightBlue]{hyperref}
\usepackage{subcaption}

\usepackage{array}

\allowdisplaybreaks

\usepackage[dvipsnames]{xcolor}
\usepackage{enumerate,enumitem}
\usepackage{csvsimple,booktabs}
\usepackage{float,subfig}
\usepackage{graphicx,epsfig,xcolor}
\usepackage{multirow}

\usepackage{fullpage}

\usepackage{ebgaramond}
\usepackage[OT1]{fontenc}
\usepackage[utf8]{inputenc}

\usepackage{thmtools}
\declaretheorem{theorem}

\declaretheorem[sibling=theorem]{lemma}

\newcommand{\beq}{\begin{equation}}
\newcommand{\eeq}{\end{equation}}
\newcommand{\beqa}{\begin{eqnarray}}
\newcommand{\eeqa}{\end{eqnarray}}
\newcommand{\beqan}{\begin{eqnarray*}}
\newcommand{\eeqan}{\end{eqnarray*}}

\newcommand{\E}{\mathds{E} }

\newcommand\raiseT[2]{\raisebox{0.25ex}{$#1#2$}
}

\newcounter{l1}
\newcounter{l2}
\newcounter{l3}
\newcommand{\bdotlist}{\begin{list}{$\bullet$}{}}
\newcommand{\bboxlist}{\begin{list}{$\Box$}{}}
\newcommand{\bbboxlist}{\begin{list}{\raisebox{.005in}{{\tiny
$\blacksquare$ \ \ }}}{}}
\newcommand{\bdashlist}{\begin{list}{$-$}{} }
\newcommand{\blist}{\begin{list}{}{} }
\newcommand{\barablist}{\begin{list}{\arabic{l1}}{\usecounter{l1}}}
\newcommand{\balphlist}{\begin{list}{(\alph{l2})}{\usecounter{l2}}}
\newcommand{\bAlphlist}{\begin{list}{\Alph{l2}.}{\usecounter{l2}}}
\newcommand{\bdiamlist}{\begin{list}{$\diamond$}{}}
\newcommand{\bromalist}{\begin{list}{(\roman{l3})}{\usecounter{l3}}}

\renewcommand{\top}{{\mathpalette\raiseT\intercal}}

\title{
Harnessing Information in Incentive Design}

\author{Raj Kiriti Velicheti \qquad Subhonmesh Bose \qquad Tamer Ba\c{s}ar
\thanks{All authors are affiliated with the Department of Electrical and Computer Engineering and the Coordinated Science Laboratory at the University of Illinois Urbana-Champaign, Urbana, IL 61801. Emails: \{rkv4, boses, basar1\}@illinois.edu. This work was partially supported by the grant NSF ECCS 2349418 from the US National Science Foundation. }}

\begin{document}

\maketitle

\begin{abstract}
Incentive design deals with interaction between a principal and an agent where the former can shape the latter's utility through a policy commitment. It is well known that the principal faces an \textit{information rent} when dealing with an agent that has informational advantage. In this work,
we embark on a systematic study of the effect of information asymmetry in incentive design games. Specifically, we first demonstrate that it is in principal's interest to decrease this information asymmetry. To mitigate this uncertainty, we let the principal gather information either by letting the agent shape her belief (aka \textit{Information Design}), or by paying to acquire it. Providing solutions to all these cases we show that while introduction of uncertainty increases the principal's cost, letting the agent shape its belief can be advantageous. We study information asymmetry and information acquisition in both matrix games and quadratic Gaussian game setups.    
\end{abstract}
\newcommand{\PPP}{{$P$}}
\newcommand{\AAA}{{$A$}}










\section{Introduction}


In a variety of principal-agent interactions, a principal (\PPP{}) seeks to induce a desired response from an agent (\AAA{}), where the two parties may have different strategic objectives, e.g., see \cite{bacsar2024inducement}. Incentive design is a widely studied mechanism for inducement. It defines a hierarchical decision-making problem, where \PPP{} announces how she would react to \AAA{}'s action. \AAA{} then acts, after which \PPP{} implements her announced strategy. These interactions can be studied via feedback Stackelberg game setups. The feedback structure allows the principal to commit to a \emph{reaction} policy to \AAA{}'s action instead of an open-loop action policy at the start of the game.

A long line of work has studied stochastic incentive design games where the two parties may have different information about cost-relevant random variables, e.g., see \cite{bacsar1984affine,cansever1982minimum,zheng1982existence,salman1981incentive,8b51acce26a144b5949c35a7e184ada1,cansever1985optimum}. However, a common assumption in these setups has been that \PPP{} knows more than \AAA{}. Indeed, if the goal is to influence \AAA{} to act exactly how \PPP{} would have preferred \AAA{} to act to minimize the former's cost, it is conceivable that \PPP{} must know what \AAA{} knows. However, in a variety of contexts, such an assumption does not hold. 
Consider for instance the interaction between an insurance provider and a policyholder. Although risk-sharing is established upon signing a contract (the incentive mechanism) where the insurance provider seeks to influence the policyholder's actions through offering premium discounts for compliance with specific rules, the policyholder typically has more knowledge about their personal circumstances and risk-preference than the insurer does. In this paper, we consider an incentive design setup, where we allow \AAA{} to have \emph{informational advantage} and then study ways in which \PPP{} can harness information to better her performance from playing a Bayesian information design game with \AAA{}.




To concretely describe our setup, consider a baseline incentive design game $\mathfrak{G}_1$ between \PPP{} and \AAA{}, where \PPP{} takes action $u \in \mathcal{U}$ and \AAA{} takes action $v \in \mathcal{V}$, which yield them costs, $c_P(u,v; \theta)$ and $c_A(u,v; \theta)$, respectively. Here, $\theta \in \Theta$ is a cost-relevant random variable (often referred to as the state of the world) for which both players hold the common prior $\mu_0 \in \Delta(\Theta)$, the probability distribution over $\Theta$. Game $\mathfrak{G}_1$ proceeds as follows. 
\PPP{} first commits to a policy $\gamma(v;\theta)$ of how she will react to the agent's action $v$, when the state is $\theta$. Such a policy can constitute an action in $\mathcal{U}$ or a mixed strategy as a distribution over $\mathcal{U}$. 
Then, the state $\theta$ is realized and \AAA{} observes $\theta$, takes an action $v \in \mathcal{V}$, and \PPP{} reacts to $v$ via the committed policy $\gamma(v, \theta)$. This game is of feedback Stackelberg-type, where instead of choosing a strategy to first act, \PPP{} commits to a policy to react to an action taken by \AAA{}. The committed policy can be thought of as a rule-book that \PPP{} proposes to induce a desired response from \AAA{}. Such games have applications in the design of tax codes and environmental taxes, among others. See \cite{groves1973incentives,bacsar1998dynamic,fudenberg1991game} for a survey. 

The Stackelberg equilibrium in $\mathfrak{G}_1$ can be described as follows. \PPP{} solves for $\gamma^\star(\cdot; \theta)$ via
\begin{align}
    \underset{\gamma(\cdot; \theta)
    }{\text{minimize}} \;  \E[c_P(\gamma(v^\star(\gamma;\theta); \theta), v^\star(\gamma;\theta); \theta)],
\end{align}
where \AAA{} solves for $v^\star(\gamma;\theta)$ through
\begin{align}
    \underset{v(\gamma; \theta) 
    }{\text{minimize}} \; \E[c_A( \gamma(v; \theta), v; \theta)],
\end{align}
assuming $v^\star$ is unique.\footnote{We sidestep problems introduced by non-unique responses from \AAA{} throughout and relegate a thorough investigation of the same to future work.}
In this game, \PPP{} and \AAA{} have access to the exact realization of $\theta$ and hence can condition their policy on it. Our object of study in this paper entails variants of $\mathfrak{G}_1$, where \PPP{} and \AAA{} have different information about $\theta$. Specifically, we consider three different classes of games, $\mathfrak{G}_2$, $\mathfrak{G}_3$, and $\mathfrak{G}_4$. In $\mathfrak{G}_2$, we let \AAA{} know the realization of $\theta$, but not \PPP{}. While in this game we retain the incentive design advantage for \PPP{}, we give informational advantage to \AAA{}. Not knowing $\theta$, \PPP{} can only solve a Bayesian version of a feedback-Stackelberg equilibrium. As one would expect, \PPP{} will incur a higher (expected) cost in an equilibrium of the Bayesian game compared to that in $\mathfrak{G}_1$. The difference is indeed the price \PPP{} pays for not knowing $\theta$ in this game of asymmetric information.

Can \PPP{} acquire information about $\theta$ and garner better costs on an equilibrium path? We address this question in this paper, by considering two different mechanisms of information acquisition in $\mathfrak{G}_3$ and $\mathfrak{G}_4$. In $\mathfrak{G}_3$, we allow \AAA{} to report a signal that is correlated with $\theta$ and let \PPP{} update her belief about $\theta$. However, \AAA{} can act strategically in how it designs the signal to persuade \PPP{}. In a sense, $\mathfrak{G}_3$ embodies an example where \emph{Bayesian persuasion} from \cite{kamenica2011bayesian} meets incentive design games in \cite{bacsar1984affine}. Here, the information-rich player (\AAA{}) sends a signal and seeks to influence how the information-poor player (\PPP{}) would design its incentive policy. As will become clear from our examples, \AAA{} acting purely out of self-interest can choose to reveal information about $\theta$ to \PPP{}, allowing \PPP{} to incur a cost lower than in $\mathfrak{G}_2$. Our setup is more complex than usual Bayesian persuasion problems considered in \cite{kamenica2011bayesian} because we accommodate incentive design as the reaction of the receiver, which in our problem is \PPP{}.

In $\mathfrak{G}_4$, we consider a setup where \PPP{} invests in devising an experiment or a signaling mechanism. This mechanism allows \PPP{} to refine its knowledge about $\theta$ prior to interacting with \AAA{}. However, the experiment is not free and \PPP{} must pay more for an experiment that reveals \emph{more information} (suitably defined) about $\theta$. While our analysis of $\mathfrak{G}_4$ bears similarity to  the costly Bayesian persuasion problem studied in \cite{gentzkow2014costly}, our setup is fundamentally different in that the information-poor player (\PPP{}) invests in an experiment design and not the other way around.

Our key contribution is the formulation and study of information acquisition in a principal-agent interaction, where \PPP{} enjoys the incentive design advantage, while \AAA{} possesses information advantage. Before delving into our analysis, we further contextualize this study within the long literature on games of asymmetric information that has a long history dating back to \cite{aumann1995repeated}. Within the context of incentive design, the role of asymmetric information in open-loop Stackelberg games has been addressed in \cite{xu2016signaling}. To our knowledge, this is the first work that studies information asymmetry and information acquisition in feedback Stackelberg matrix games--a type of interaction that more accurately models rule-book design. While Bayesian persuasion itself has been studied in matrix game contexts in \cite{kamenica2011bayesian,dughmi2016algorithmic}, this work attempts to merge two different types of inducements--belief shaping and incentive design. In addition, the last game $\mathfrak{G}_4$ borrows tools and techniques from the persuasion literature to address a question where \PPP{} chooses to invest in shaping her own belief before she interacts with \AAA{}.

Departing from the world of finitely many actions and states, we consider the same set of questions $\mathfrak{G}_1$--$\mathfrak{G}_4$ in quadratic Gaussian (QG) games. In such games, direct computation of a Stackelberg equilibrium in stochastic incentive design problems have proven challenging. An indirect solution method leveraging a suitable team problem has been studied in \cite{bacsar1984affine}. However, due to the non-classical information structure in problems we consider, this route is rendered inapplicable. As a result, our study of Bayesian variants with information acquisition techniques in QG games with a suitable restricted class of incentive policies stands out as an important contribution of this work. Again, while Bayesian persuasion in QG games have received recent attention, e.g., in \cite{tamura2018bayesian,sayin2020persuasion,velicheti2025value}, the combination of persuasion and incentive design in QG games remains unexplored. These games provide a powerful modeling framework for inducement with continuous state/action spaces. Furthermore, the incentive policies we consider are \emph{soft} in that they depend continuously on \AAA{}'s response\cite{bacsar2024inducement}. In practical contexts, such soft policies are typically more defensible than threat policies.

Within the context of Principal-agent interactions, the effect of asymmetric information was studied under the name of \textit{adverse selection} starting from \cite{mussa1978monopoly,akerlof1978market,milgrom1990rationalizability} with follow-up works advocating reduction of such information asymmetry through phenomena such as signaling \cite{spence1978job}. For a more detailed exposition of this literature see, \cite{bolton2004contract}. However, in contrast to these works we deal with a different kind of information disclosure. In our work \AAA{} can \textit{commit} to a disclosure mechanism similar to Bayesian Persuasion \cite{kamenica2011bayesian}. Recently, \cite{roesler2017buyer,bergemann2015limits,ali2023consumer} have addressed such verifiable disclosure in principal-agent interactions similar to our work, concentrating on specific cost structures such as monopolistic pricing. \cite{bergemann2022screening} study the setting where \PPP{} has both information advantage and can can design quality differentiated products.



\cite{bergemann2022screening} study the setting where \PPP{} has both information advantage and can can design quality differentiated products. 

The paper is organized as follows.
We begin by formalizing the games $\mathfrak{G}_2$--$\mathfrak{G}_4$ in Section \ref{sec:G2.G4}. We study $\mathfrak{G}_1$--$\mathfrak{G}_4$ in matrix games in Section \ref{Sec: MatrixGames}, followed by that in QG games  in Section \ref{sec:QG}. We conclude the paper in Section \ref{sec:conc}.



\section{Formalizing Incentive Design Games with Information Asymmetry and Acquisition}
\label{sec:G2.G4}
Borrowing the notation from $\mathfrak{G}_1$, we begin by presenting the Bayesian feedback Stackelberg game $\mathfrak{G}_2$. 
In $\mathfrak{G}_2$, \PPP{} does not know $\theta$, but only knows
$\mu_0 \in \Delta(\Theta)$, while \AAA{} sees the realization of $\theta$. In this case, the principal loses the ability to make the incentive policy contingent on $\theta$, but rather commits to a policy of the form, $\gamma(v)$ that is a pure or mixed strategy to react to the action of \AAA{}. Given such a policy, \AAA{} responds by solving
\begin{align}
    \underset{v(\gamma; \theta)}{\textrm{minimize}} \; \E[c_A( \gamma(v), v; \theta)],
    \label{eq:F.resp.G2}
\end{align}
where the expectation is taken with respect to any randomness considered in $\gamma$. \PPP{} then solves $\gamma^\star$ via
\begin{align}
     \underset{\gamma(\cdot)}{\textrm{minimize}} \; \E\left[ c_P(\gamma(v^\star(\gamma,\theta)), v^\star(\gamma,\theta); \theta) \right],
    \label{eq:L.opt.G2}
\end{align}
where the expectation is computed with respect to possible randomness in $\gamma$ and ${\theta \sim \mu_0}$.

Next, we investigate two different variations of $\mathfrak{G}_2$ that allow different forms of information acquisition by \PPP{}.
In $\mathfrak{G}_3$, \AAA{} designs a signaling policy $\pi$ and commits to it before seeing the state. Then, he observes the state $\theta$ and generates the signal $s$ that takes values in a suitable signal space $\mathcal{S}$ based on the realized $\theta$ and the signal likelihoods encoded in the signaling mechanism $\pi$. \PPP{} observes the signal $s$ and updates the belief from $\mu_0$ to the posterior distribution $\mu_s(s, \pi) \in \Delta(\Theta)$ via Bayes' rule. Since $s \sim \pi$ is random, the posterior distribution $\mu_s(s, \pi)$ is consequently random. A signaling mechanism $\pi$ thus induces a distribution $\tau(\mu_0, \pi)$ over posterior beliefs. 

For a given choice of a signaling mechanism $\pi$ by \AAA{} in $\mathfrak{G}_3$,  \PPP{} then solves \eqref{eq:L.opt.G2}, where the expectation is now computed with respect to the randomness in $\gamma$ and $\theta \sim \mu_s(s, \pi)$. \AAA{} must \emph{shape the belief} of \PPP{} in a way that \PPP{} prescribes an incentive policy that in turn is conducive for \AAA{}. That is, \AAA{} must solve \eqref{eq:F.resp.G2} over both $v(\gamma; \theta)$ and $\pi$, knowing that \AAA{} will respond by solving  \eqref{eq:L.opt.G2} with the shaped belief $\mu_s(s, \pi)$ about $\theta$. The expectation in solving \eqref{eq:F.resp.G2} is computed with respect to randomness in the response $\gamma$, state $\theta \in \mu_0$, and the signal $s \in \pi$. Game $\mathfrak{G}_3$ thus combines Bayesian persuasion by \AAA{} to influence the incentive design by \PPP{}.

Finally, we consider game $\mathfrak{G}_4$, where \PPP{} invests in designing an information channel $\pi$ that generates a signal $s\in \mathcal{S}$ that carries information about $\theta$. Then, the state is realized and observed by \AAA{}, but not \PPP{}. However, \PPP{} uses $\pi$ to generate a signal $s$ and uses $s$ to update her belief about $\theta$. Akin to the setup in $\mathfrak{G}_3$, \PPP{} updates her belief from $\mu_0$ to the posterior distribution $\mu_s(s, \pi)$ and proceeds with the design of an incentive mechanism for \AAA{} with this updated belief. Thus, $\mathfrak{G}_4$ is similar to $\mathfrak{G}_3$, except that \PPP{} must now design the signaling experiment to her advantage instead of leaving it to \AAA{}, whose incentives might not be aligned with \PPP{}.


If experimentation or channel construction is free, \PPP{} will always choose a channel that reveals $\theta$ exactly to circumvent paying any price (in terms of achievable equilibrium costs) for not knowing $\theta$. Typically, higher quality experiments are costlier to design. Myriad of ways exist to model experimentation costs, e.g., see \cite{matyskova2023bayesian}. We follow the model in \cite{gentzkow2014costly} and consider the cost of channel construction to be the reduction in entropy from a signaling mechanism as
\begin{align}
    c(\pi) := \E[H(\mu'_0)-H(\mu'_s(s, \pi))],
    \label{eq:cost.pi}
\end{align}
where $\mu'_0\in\Delta(\Theta)$ is a reference prior that $\pi$ maps to a posterior distribution $\mu'_s(s, \pi)$ when $\pi$ generates the signal $s$.  Call the distribution over these posteriors $\tau(\mu'_0, \pi) \in\Delta(\Delta(\Theta))$.
The expectation in the relation above is computed with respect to the randomness in $s \sim \pi$, $\theta \sim \mu_0$, and $\mu'_0$.
One can take $\mu'_0$ to be the prior $\mu_0$, but such a choice makes the channel cost prior-dependent, a property we avoid as in \cite{gentzkow2014costly}. Rather, we choose $\mu'_0$ according to the problem at hand. 

With this information channel construction cost model, \PPP{} in $\mathfrak{G}_4$ therefore solves \eqref{eq:L.opt.G2} over both the incentive mechanism $\gamma$ and the information channel design $\pi$, where the objective function in \eqref{eq:L.opt.G2} is augmented with $c(\pi)$. \AAA{} then responds by solving \eqref{eq:F.resp.G2} similar to that in $\mathfrak{G}_2$. In what follows, we study $\mathfrak{G}_1$--$\mathfrak{G}_4$ in different game setups.

\section{Studying $\mathfrak{G}_1$--$\mathfrak{G}_4$ in Matrix Games}\label{Sec: MatrixGames}
To study incentive design with information asymmetry and acquisition via  $\mathfrak{G}_1$--$\mathfrak{G}_4$ in matrix games, consider a principal-agent interaction with action sets $\mathcal{U} = \{u_1, \ldots, u_m\}$, $\mathcal{V} = \{v_1, \ldots, v_n\}$, and state space $\Theta = \{\theta_1, \theta_2\}$.
The costs of \PPP{} and \AAA{} for $\theta \in \Theta$ can be represented via  matrices $C_P(\theta), C_A(\theta)\in\mathbb{R}^{m\times n}$, where
\begin{align}
c_P(u_i,v_j ;\theta)= [C_P(\theta)]_{i,j},  c_A(u_i,v_j ;\theta)=[C_A(\theta)]_{i,j},
\end{align}
for each $\theta \in \Theta$,
making \PPP{} and \AAA{} choose the row and column of the matrices as actions, respectively. The common prior is given by $\mu_0 \in \Delta(\Theta)$, the probability simplex over $\Theta$. With this notation, we now study $\mathfrak{G}_1$--$\mathfrak{G}_2$ in Section \ref{Subsec: g1_g2_matrix} and $\mathfrak{G}_3$--$\mathfrak{G}_4$ in Section \ref{Subsec: g3_g4_matrix}, respectively.

\vspace{-0.25 cm}
\subsection{Solving $\mathfrak{G}_1$ and $\mathfrak{G}_2$ via Linear Programming}\label{Subsec: g1_g2_matrix}
\vspace{-0.15 cm}
When the state is $\theta_k \in \Theta$, let the response of \PPP{} to \AAA{} playing $v_j \in \mathcal{V}$ be given by the mixed strategy $\gamma(v_j; \theta) \in \Delta(\mathcal{U})$. To identify a Stackelberg equilibrium, we compute the least cost that \PPP{} can obtain when \AAA{} best responds with $v_j$, and then finds the minimum over those costs as $v_j$ varies over $\mathcal{V}$.
In other words, 
\PPP{} obtains its optimal cost along an equilibrium path by solving
\begin{align}
    J_{P,1}^\star(\theta_k) := \underset{v_j \in \mathcal{V}}{\text{minimum}} \; \gamma^\star(v_j;\theta_k)^\top C_P(\theta_k)e_j,
    \label{eq:J1.v}
\end{align}
where $e_j$ is a vector with all zeros except unity in the $j^{\textrm{th}}$ position, and $\GS(v_j;\theta_k)$ is an optimal solution of the following linear program.
\begin{alignat}{2}\label{eq: g1_mat}
\begin{aligned}
&  \underset{\gamma(\cdot; \theta_k) \in \Delta(\mathcal{U})}{\text{minimize}} && \gamma(v_j;\theta_k)^\top C_P(\theta_k)e_j,\\
& \; \text{subject to}&& \gamma(v_j;\theta_k)^\top C_A(\theta_k)e_j\leq \gamma(v_l;\theta_k)^\top C_A(\theta_k) e_l,\\
& \; &&
 \forall l\in [n].
 \end{aligned}
\end{alignat}
With a slight abuse of notation, we define $J_{P,1}^\star(\mu_0) := \E_{\theta_k\sim \mu_0}[J_{P,1}^\star(\theta_k)]$.
The constraint encodes \AAA{} playing $v_j$ under the policies $\gamma(\cdot; \theta_k)$ chosen by \PPP{}.\footnote{Problem \eqref{eq: g1_mat} solves for a feedback Stackelberg equilibrium. An open-loop Stackelberg equilibrium can be computed by imposing the additional constraint $\gamma(v_j;\theta_k)=\gamma(v_l;\theta_k), \forall j,k\in[n]$ \cite{conitzer2006computing}.} This amounts to \PPP{} solving $m$ linear programs over $mn$ variables. If the optimal action of \AAA{} from \eqref{eq:J1.v} is $v_{j^\star}(\theta_k)$, then define $J_{A,1}^\star (\mu_0):= \E_{\theta_k\sim \mu_0}[\gamma^\star(v_{j^\star}(\theta_k), \theta_k)^\top C_A(\theta_k) e_{j^\star}]$ as \AAA{}'s optimal cost.

Consider an example matrix game with $m=n=2$ and
\begin{gather}
\begin{gathered}
    C_P(\theta_1) = \begin{pmatrix}
        5 & 5 \\ 5 & 1
    \end{pmatrix},
    \;
    C_P(\theta_2) = \begin{pmatrix}
        5 & 1 \\ 5 & 5
    \end{pmatrix},
    \\
    C_A(\theta_1) = \begin{pmatrix}
        4 & 3 \\ 2 & 3
    \end{pmatrix},
    \;
    C_A(\theta_2) = \begin{pmatrix}
        2 & 3 \\ 4 & 2
    \end{pmatrix}.
\end{gathered}
\label{eq:ex.costs}
\end{gather}
Let $\mu_0(\theta_1)=0.4$. Utilizing \eqref{eq:J1.v}--\eqref{eq: g1_mat}, we obtain 
\begin{gather}
    \begin{gathered}
\gamma^\star(v_1;\theta_1)=\begin{pmatrix}
    0.5 & 0.5
\end{pmatrix}^\top,\gamma^\star(v_2;\theta_1)=\begin{pmatrix}
    0 & 1
\end{pmatrix}^\top, 
\\
\gamma^\star(v_1;\theta_2)=\begin{pmatrix}
    0 & 1
\end{pmatrix}^\top, \gamma^\star(v_2;\theta_2)=\begin{pmatrix}
    1 & 0
\end{pmatrix}^\top,
\\ 
    v_{j^\star}(\theta_1) = v_{j^\star}(\theta_2) = v_2,J_{P,1}^\star(\mu_0) = 1, 
    J_{A,1}^\star(\mu_0) = 3.
    \end{gathered}
\end{gather}

We next present a solution technique for $\mathfrak{G}_2$, where the state is sampled form $\mu_0$, but the realization is only known to \AAA{} and not \PPP{}. However, the latter knows $\mu_0$ and hence solves a Bayesian variant of $\mathfrak{G}_1$. Here, \PPP{} finds the least cost that she can obtain when \AAA{} responds by playing the pair $(v_i, v_j)$ with states $(\theta_1,\theta_2)$, respectively, and then minimizes that cost over all possible response pairs in $\mathcal{V} \times \mathcal{V}$. Thus, \PPP{} solves her equilibrium cost as follows.
\begin{align}\label{eq: g2_mat_min}
\begin{split}
J_{P,2}^\star(\mu_0) &= \underset{i,j \in [m]^2}{\text{minimum}} \{\mu_0(\theta_1)\gamma^\star_{ij}(v_i)^\top C_P(\theta_1)e_i\\
   & \qquad +\mu_0(\theta_2)\gamma^\star_{ij}(v_j)^\top C_P(\theta_2)e_j\},
\end{split}
\end{align}
where $\GS_{ij}$ is a minimizer of
    \begin{alignat}{2}
    \label{prog: g2_mat}
    \begin{aligned}
        &\underset{\gamma(\cdot) \in \Delta(\mathcal{U})}{\text{minimize}} && \mu_0(\theta_1) \gamma(v_i)^\top C_P(\theta_1)e_i 
        \\
        &&& \quad + \mu_0(\theta_2)\gamma(v_j)^\top C_P(\theta_2)e_j,\\
     &\text{subject to}\quad && 
       \gamma(v_i)^\top C_A(\theta_1)e_i\leq \gamma(v_l)^\top C_A(\theta_1) e_l,\\
       & &&\gamma(v_j)^\top C_A(\theta_2)e_j\leq \gamma(v_{l'})^\top C_A(\theta_2) e_{l'},\\
       & &&  l \times l' \in [n]^2.
       \end{aligned}
    \end{alignat}
Cost for \AAA{} at an equilibrium is given by 
\begin{align}
\begin{aligned}
    J_{A, 2}^{\star F} (\mu_0) 
    &= \mu_0(\theta_1)\gamma_{i^\star j^\star}^{\star\top}(v_{i^\star}) C_A(\theta_1)e_{i^\star} \\
    & \; \;
    +\mu_0(\theta_2)\gamma_{i^\star j^\star}^{\star\top}(v_{j^\star})C_A(\theta_2)e_{j^\star},
\end{aligned}
\end{align}
where the minimum of \eqref{eq: g2_mat_min} is attained at $(i^\star, j^\star)$. 

In $\mathfrak{G}_2$, \PPP{} is oblivious to the realization of the state, and hence, her strategies do not depend on the same. The equilibrium costs of $\mathfrak{G}_1$ and $\mathfrak{G}_2$ satisfy
\begin{align}
    J_{P,1}^\star (\mu_0)\leq J_{P,2}^\star(\mu_0), \; k=1,2.
\end{align} 
The difference between these costs can be viewed as the price \PPP{} pays for ignorance of the state. In the example game with costs described in \eqref{eq:ex.costs}, an optimal strategy for \PPP{} in this Bayesian game $\mathfrak{G}_2$, is 
\begin{align}
    \gamma^\star(v_1) = \begin{pmatrix}
    0 & 1
\end{pmatrix}^\top, \; \gamma^\star(v_2) = \begin{pmatrix}
    1 & 0
\end{pmatrix}^\top,
\end{align}
leading to $J_{P,2}^\star = 2.6$. \AAA{} responds by playing $(v_1, v_2)$ in states $\theta = (\theta_1, \theta_2)$, respectively, achieving $J_{A,2}^\star(\mu_0) = 2.6$. Thus, \PPP{} incurs an extra cost in $\mathfrak{G}_2$ compared to that in $\mathfrak{G}_1$; the story is opposite for \AAA{}. We next study $\mathfrak{G}_3$ and $\mathfrak{G}_4$ that provide mechanisms via which \PPP{} might be able to close this performance gap by acquiring more information. 

\vspace{-0.1cm}
\subsection{Harnessing Information via $\mathfrak{G}_3$ or $\mathfrak{G}_4$}\label{Subsec: g3_g4_matrix}
Should \PPP{} allow \AAA{} to reveal information in $\mathfrak{G}_3$ about $\theta$ or invest in construction of an information channel in $\mathfrak{G}_3$ to better her equilibrium costs from that in $\mathfrak{G}_2$? In both these games, the information acquisition step changes the prior of \PPP{} from $\mu_0$ to a random posterior $\mu_s$ so that \PPP{} accrues a cost of $J_{P,2}^\star(\mu_s)$, possibly augmented with any cost associated with information acquisition. In our first result, we establish a structural property of $J_{P,2}^\star$ that implies that \PPP{} can only benefit from more information. 
\begin{lemma}\label{lem: cvx_obj}
$J_{P,2}^\star :\Delta(\Theta) \to \mathbb{R}$ is piecewise affine and concave.
\end{lemma}
\begin{proof}
    The optimization problem in \eqref{prog: g2_mat} is a linear program whose feasible set is a polytope that is independent of $\mu_0$. The corners of said polytope defines its candidate optimal solutions. Its objective function being affine in $\mu_0$, the optimal value of \eqref{prog: g2_mat} becomes the minimum of finitely many affine functions of $\mu_0$, making it piecewise affine and concave. The same argument carries over to deduce the structure of $J_{P,2}^\star$ as it takes the minimum over the aforementioned piecewise affine and concave functions, completing the proof. 
\end{proof}

Figure \ref{fig: principal_concave_cost} illustrates the piecewise affine and concave structure of $J_{P,2}^\star$ as a function of the belief about $\theta$ held by \PPP{} in our example game with cost structure \eqref{eq:ex.costs}. Notice that \PPP{} has lower costs towards full information beliefs, i.e., when $\mu_0(\theta_1)$ is close to zero or unity, and higher costs when she is more uncertain about the state, i.e., when $\mu_0(\theta_1)$ is near $1/2$, demonstrating the need for information acquisition. 
\begin{figure}[h]
    \centering
    \includegraphics[ width=0.75\linewidth]{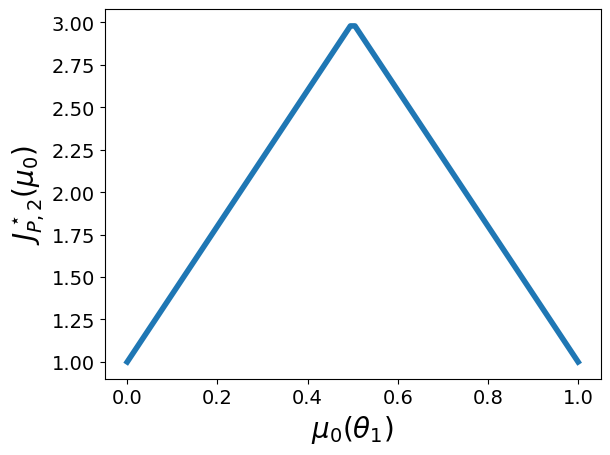}
    \caption{Principal's cost is piecewise affine and concave}
    \label{fig: principal_concave_cost}
\end{figure}

Information acquisition requires \PPP{} to observe a signal that is correlated with the state. In $\mathfrak{G}_3$, she allows \AAA{} to design a signaling scheme and in $\mathfrak{G}_4$, she invests in an experiment that generates such a signal. To study these games, let \PPP{} observe a signal $s \in \mathcal{S}$ that is generated according to likelihoods $\{\pi(s|\theta_k)\}_{k\in[2]}$, where $\mathcal{S}$ is a suitable space of signals. In $\mathfrak{G}_3$, we assume that \AAA{} commits to $\pi$ that is revealed to \PPP{}. In $\mathfrak{G}_4$, \PPP{} designs $\pi$ at a cost.

In both $\mathfrak{G}_3$ and $\mathfrak{G}_4$, \PPP{} observes $s$ and updates her belief from the prior $\mu_0$ to $\mu_s$ using Bayes' rule,
\begin{align}
    \mu_s(\theta) =\frac{\pi(s\mid\theta)\mu_0(\theta)}{\sum_k\pi (s\mid\theta_k)\mu_0(\theta_k)}.
\end{align}
Since $s$ is random, so is $\mu_s$. Hence, each signaling mechanism $\pi$ and a prior $\mu_0$ leads to a distribution over posteriors $\tau(\mu_0, \pi)\in\Delta(\Delta(\theta))$. In $\mathfrak{G}_3$ and $\mathfrak{G}_4$, after the Bayesian update, \PPP{} and \AAA{} engage in playing $\mathfrak{G}_2$ with the prior replaced by $\mu_s$. 

Bayes' rule yields the Markov property that $\E[\mu_s] = \mu_0$, where this expectation is taken w.r.t. $\tau(\mu_0, \pi)$ induced by $\pi$. Coupled with Lemma \ref{lem: cvx_obj}, we deduce that
\begin{align}
    \E[J_{P,2}^\star(\mu_s)] \leq J_{P,2}^\star(\E[\mu_s]) = J_{P,2}^\star(\mu_0),
\end{align}
meaning that information can only benefit \PPP{} in expectation.
However, this does not hold for \AAA{} who has the informational advantage and hence can shape the posterior belief $\mu_s$. A key result of \cite{kamenica2011bayesian} is that the Markov property of beliefs is not just necessary but also sufficient for any signaling mechanism. That is, any set of posteriors satisfying the Markov property $\E_{\mu_s \sim \tau(\mu_0, \pi)}[\mu_s] = \mu_0$ can be induced by a suitable $\pi$. Thus, \AAA{} in $\mathfrak{G}_3$ solves
\begin{align}\label{Eq: g3_obj}
\begin{split}
    \underset{\tau}{\text{minimize}} \ \E_{\tau}[J_{A,2}^{\star}(\mu_s)], \;
    \text{subject to }\E_{\mu_s\sim\tau}[\mu_s]=\mu_0,
\end{split}
\end{align}
where $J_{A, 2}^{\star}(\mu_s)$ is the cost to \AAA{} at an equilibrium in $\mathfrak{G}_2$, when the belief held by \PPP{} is $\mu_s$. Let the optimal cost of the above problem be given by $J_{A, 3}^\star(\mu_0)$.

For an arbitrary function $\phi : \mathcal{X} \to \mathbb{R}$ from an arbitrary subset $\mathcal{X}$ of a Euclidean space, let ${\sf co} [\phi](y) := \sup_{x\in\mathcal{X}} \left(y^\top x - \phi(x)\right)$ denote its Fenchel conjugate.  Then, the bi-conjugate of a function $\phi$ is ${\sf co}[{\sf co}[\phi]]$. An immediate consequence of \cite[Corollary 2]{kamenica2011bayesian} is the following result, the proof of which is omitted.
\begin{lemma}\label{lem: g2_convexification}
    $J_{A,3}^\star(\mu_0) = {\sf co}[{\sf co}[J_{A,2}^\star]](\mu_0).$
\end{lemma}
Graphically, the epigraph of the bi-conjugate of $J_{A,2}^\star$ is the convex hull of the epigraph of $J_{A,2}^\star$. 
This relationship follows from the fact that if $(\mu_0, t)$ lies in the epigraph of $J_{A,2}^\star$, then there exists a distribution $\tau$ such that $\E_{\mu_s \sim \tau}[\mu_s]=\mu_0$ and $\E[J_{A,2}^\star(\mu_s)] = t$. Said differently, the bi-conjugate \emph{convexifies} $J_{A,2}^\star$. If that function is convex to begin with, there is no advantage in inducing beliefs in \PPP{} via signaling. Otherwise, a non-zero difference between $J_{A,2}^\star$ and its bi-conjugate at $\mu_0$ equals the benefit from persuasion.

In general, computing the bi-conjugate of a function is challenging. However, we are able to compute the solution for $\mathfrak{G}_3$ using the observation in \cite{kamenica2011bayesian} that it is sufficient to consider a signal space with cardinality that is the minimum between the number of states (which for our problem is $|\Theta| = 2$), and the number of actions available to \PPP{} (which is finite for our problem as we describe next). By sufficient, we mean that with such signal spaces, it is possible to construct a signaling mechanism to induce \emph{any} distribution $\tau$ over possible posteriors that satisfies the Markov property. To see why the effective number of actions for \PPP{} is finite, consider \eqref{prog: g2_mat}.
The feasible set of this problem for a specific $(v_i, v_j)$ is a polytope that is independent of $\mu_0$. Enumerate the vertices of this polytope as $\Xi_{ij}$. For any $\mu_0$, \PPP{} solves $n^2$ linear programs over these polytopes corresponding to all possible $(v_i, v_j)$ pairs, each of which admits corner solutions. In other words, the optimal incentive policy for \PPP{} is one of the vertices of this collection of $n^2$ polytopes, given by $\Xi := \cup_{(i,j) \in [n]^2} \Xi_{ij}$. These finite number of vertices can be thought of as the candidate actions for \PPP{}. Then, \AAA{} \emph{recommends} one of the ``actions'' in $\Xi$ through a (mixed) signaling policy $\pi(\cdot \mid \theta_k)$ such that it is incentive compatible for \PPP{} to accept the recommendation from \AAA{} and in doing so, \AAA{} minimizes his expected cost when \PPP{} follows that recommendation. That is, $J_{A,3}^\star(\mu_0)$ becomes the optimal cost of the following problem:
\begin{alignat}{2}
\begin{aligned}
        &\underset{\pi}{\text{minimize}} &&\sum_{i,j,\xi}\pi(\gamma^\xi_{ij}\mid \theta_1)\mu_0(\theta_1)\gamma^\xi_{ij}(v_i)C_A(\theta_1)e_i\\
& &&+\pi(\gamma^\xi_{ij}\mid \theta_2)\mu_0(\theta_2)\gamma^\xi_{ij}(v_j)C_A(\theta_2)e_j,\\
&\text{subject to }&&\sum_{i,j,\xi}\pi(\gamma^\xi_{ij}\mid \theta)=1, \pi(\gamma^\xi_{ij}\mid\theta_k)\geq 0,\\
& &&\pi(\gamma^\xi_{ij}\mid \theta_1)\mu_0(\theta_1)\gamma^\xi_{ij}(v_i) C_P(\theta_1)e_i\\
& && + \pi(\gamma^\xi_{ij}\mid \theta_2)\mu_0(\theta_2)\gamma^\xi_{ij}(v_j) C_P(\theta_2)e_j\\
& && \leq \pi(\gamma^\xi_{ij}\mid \theta_1)\mu_0(\theta_1)\gamma_{i' j'}^{\xi'}(v_{i'}) C_P(\theta_1)e_{i'}\\
& && + \pi(\gamma^\xi_{ij}\mid \theta_2)\mu_0(\theta_2)\gamma_{i' j'}^{\xi'}(v_{j'}) C_P(\theta_2)e_{j'}\\
& &&\forall \xi  \in \Xi_{ij}, (i,j) \in [n]^2, 
\\
&&&
\xi'  \in \Xi_{i' j'}, (i',j') \in [n]^2.
\end{aligned}
\label{prog: g3_agent}
\end{alignat}
We remark that while computation of bi-conjugate of an arbitrary function is often difficult, it is not so for piecewise affine functions--a property that underlies the construction of the bi-conjugate of $J_{A,2}^\star$ through the above optimization problem.

Consider $\mathfrak{G}_3$ with cost matrices in \eqref{eq:ex.costs}. The vertices of $\Xi_{ij}$ for $(i,j)\in[2]^2$ are given by
\begin{align}
\begin{aligned}
    \Xi_{1,1} &= 
    \left\{ \begin{pmatrix} 
    0.5 & 1\\
    0.5 & 0
    \end{pmatrix} \right\},
    \\
    \Xi_{1,2} &= 
    \left\{ \begin{pmatrix} 
    0 & 0\\
    1 & 1
    \end{pmatrix}, \begin{pmatrix} 
    0 & 1\\
    1 & 0
    \end{pmatrix}, \begin{pmatrix} 
    0.5 & 1\\
    0.5 & 0
    \end{pmatrix} \right\},
    \\
    \Xi_{2,1} &= 
    \left\{ \begin{pmatrix} 
    0.5 & 1\\
    0.5 & 0
    \end{pmatrix}, \begin{pmatrix} 
    1 & 0\\
    0 & 1
    \end{pmatrix} \right\},
    \\
    \Xi_{2,2} &= 
    \left\{ \begin{pmatrix} 
    0.5 & 1\\
    0.5 & 0
    \end{pmatrix}, \begin{pmatrix} 
    1 & 0\\
    0 & 1
    \end{pmatrix} \right\}.
\end{aligned}
\end{align}
Each entry of $\Xi_{ij}$ is a $2\times 2$ matrix, whose $i^{\textrm{th}}$ column defines the mixed strategy $\gamma(v_i)$ for \PPP{} to respond when \AAA{} chooses $v_i$. There are four unique $\gamma$'s among $\Xi$'s. One can verify that the objective function of \eqref{prog: g3_agent} is the same when $\pi$ uniformly randomizes across any two among $\Xi_{11} \cup \ldots \cup \Xi_{22}$ and the minimum of \eqref{prog: g3_agent} indeed occurs at such a signaling policy $\pi$. 
Such a $\pi$ corresponds to the case that \AAA{} chooses \emph{not} to reveal any information about the state to \PPP{}. As a result, \PPP{} fails to better her equilibrium costs in $\mathfrak{G}_3$ from that in $\mathfrak{G}_2$. \AAA{} achieves the same costs as in $\mathfrak{G}_2$ as well. In fact, using \eqref{prog: g3_agent} one can show that 
\begin{align*}
   J^\star_{A,3}(\mu_0)= \begin{cases}
    3-\mu_0(\theta_1),  & \text{if }\mu_0(\theta_1) \leq 0.5,
    \\ 2+\mu_0(\theta_1), & \text{otherwise},
    \end{cases}
\end{align*}
which is convex, and hence, equal to its bi-conjugate.

Consider another example of $\mathfrak{G}_3$ with  cost matrices
\begin{gather}
\begin{gathered}
    C_P(\theta_1) = \begin{pmatrix}
        4.95 & 5 \\ 5 & 1
    \end{pmatrix},
    \;
    C_P(\theta_2) = \begin{pmatrix}
        5 & 1 \\ 5 & 5
    \end{pmatrix},
    \\
    C_A(\theta_1) =\begin{pmatrix}
        2 & 5\\
        1 & 2
    \end{pmatrix}, \; \; 
    C_A(\theta_2) =\begin{pmatrix}
        3 & 1\\
        1 & 2
    \end{pmatrix}.
\end{gathered}
\label{eq:ex.costs1}
\end{gather}

For this problem, the equilibrium cost for \AAA{} as a function of the induced belief (characterized by $\mu_s(\theta_1)$) and its convex envelope is drawn in Figure \ref{fig: convexification_agents_cost}.
\begin{figure}[h]
    \centering
    \includegraphics[scale =0.28]{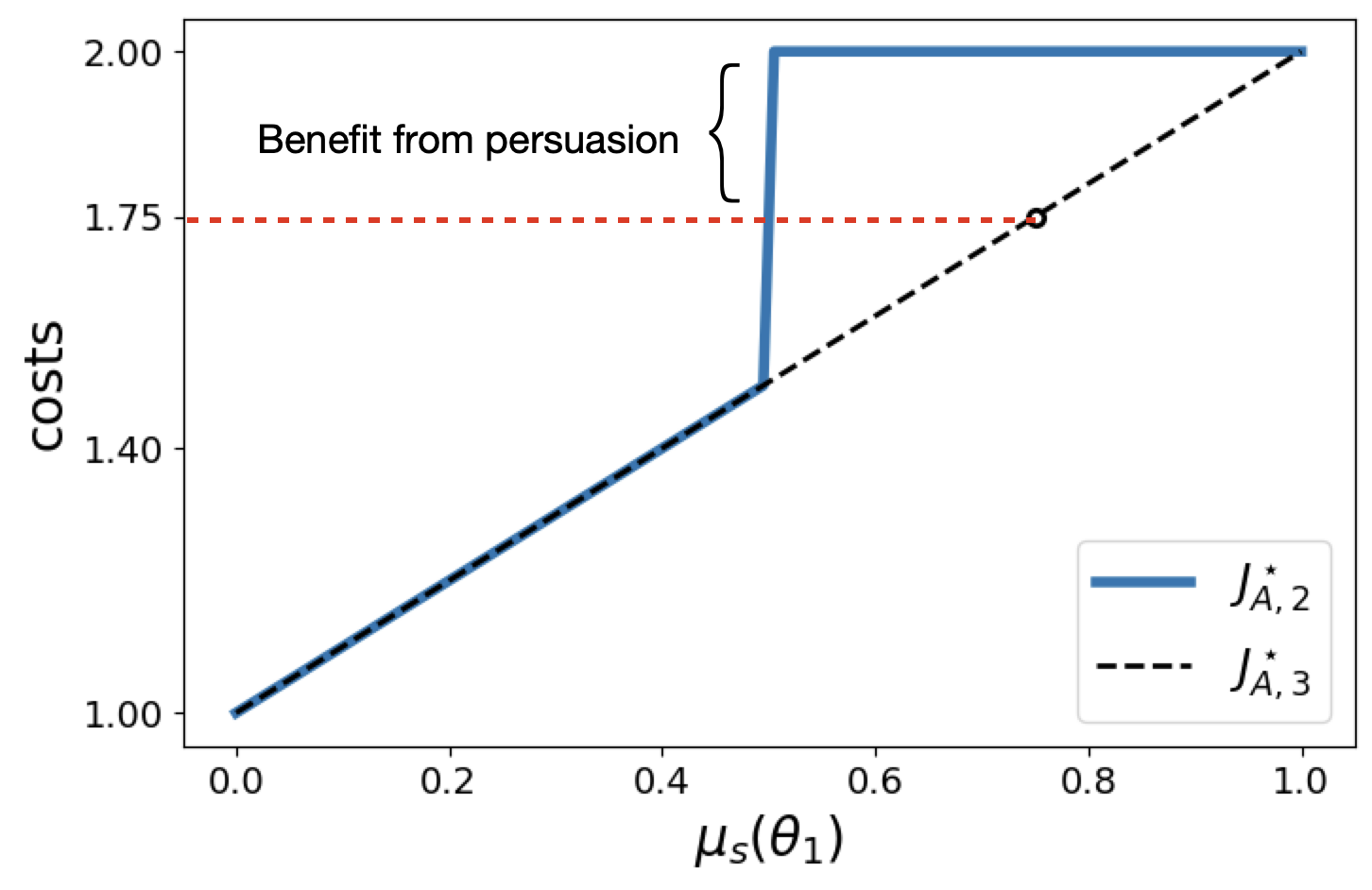}
    \caption{Illustration of convexification of \AAA{}'s cost.}
    \label{fig: convexification_agents_cost}
\end{figure}
The diagram reveals that for priors with probabilities of $\theta_1$ less than $1/2$, \AAA{} does not benefit from revealing any information about the state. As a result, \PPP{} cannot rely on \AAA{} to better its cost than in $\mathfrak{G}_2$ for such priors. However, \AAA{} stands to gain by revealing information when $\theta_1$ is more probable than $\theta_2$. For example, if $\mu_0(\theta_1) = 0.75$, the optimal signaling scheme from \eqref{prog: g3_agent} yields the candidate posterior distributions with $\mu_s(\theta_1)=0$ and $\mu_s(\theta_1)=1$ with probabilities $0.25$ and $0.75$, respectively. Then, we have $J_{A,3}^\star(\mu_0) = 1.75$. For this example,  $J_{A,2}^\star(\mu_0) = 2$, implying a gain of $0.25$ from information revelation. Finally, for this example, $J_{P,2}^\star(\mu_0) = 2$ and $J_{P,3}^\star(\mu_0) = 1$, meaning that revelation by \AAA{} also in turn benefits \PPP{}. Notice that the posterior beliefs exactly reveal the state, meaning that in this case, \PPP{} attains the cost in $\mathfrak{G}_1$ by allowing \AAA{} to signal about the state.



Next, we study $\mathfrak{G}_4$ where \PPP{} has access to a costly monitoring channel. Precisely, \PPP{} chooses the signaling mechanism $\pi$ instead of leaving it to \AAA{}, but pays for the quality of $\pi$. Consider the cost of channel construction as in \eqref{eq:cost.pi} with the reference prior $\mu'_0\in\Delta(\Theta)$ being uniform over $\Theta$ with $H(\mu'_0) = 1$. The distributions of the posteriors under $\pi$ with the priors being $\mu_0$ and $\mu'_0$ are then given by $\tau(\mu_0, \pi)$ and $\tau(\mu'_0, \pi)$, respectively.
With this notation, \PPP{} in $\mathfrak{G}_4$ minimizes
\begin{align}\label{eq: g4_agent_min}
    \E_{\tau(\mu_0, \pi)}[J_{P,2}^\star(\mu_s)] + \kappa\left( 1 - \E_{\tau(\mu'_0, \pi)}[H(\mu'_s)] \right)
\end{align}
over $\pi$, the minimum value of which we call $J_{P,4}^\star(\mu_0)$. In what follows, we identify $J_{P,4}^\star(\mu_0)$ as the bi-conjugate of a function related to $J_{P,2}^\star(\mu_0)$.

To approach the minimization, we express \eqref{eq: g4_agent_min} in terms of $\mu_s$ using the transformation,
\begin{align*}
    \mu_s'(\theta_k) = \mu_s(\theta_k) \frac{\mu_0'(\theta_k)/\mu_0(\theta_k)}{\sum_{k'}\mu_s(\theta_{k'})\mu_0'(\theta_{k'})/\mu_0(\theta_{k'})},
\end{align*}
where $\mu'_0(\theta_k) = 1/2$ for $k=1,2$.
This relation can be derived as in \cite{alonso2016bayesian} by using Bayes' rule that maps $\mu_0$ to $\mu_s$ and $\mu'_0$ to $\mu'_s$ via the common signaling scheme $\pi$. Using the above relation, we can then write 
\begin{align}
    \E_{\tau(\mu'_0, \pi)}[H(\mu'_s)] = \E_{\tau(\mu_0, \pi)}[\widetilde{H}(\mu_s)]
\end{align}
Plugging the above relation in \eqref{eq: g4_agent_min}, we then obtain the following result, similar to Lemma \ref{lem: g2_convexification}.
\begin{lemma}
$J_{P,4}^\star(\mu_0) = {\sf co}[{\sf co}[ J_{P,2}^\star - \kappa \widetilde{H}]] (\mu_0) + \kappa$.
\end{lemma}
Recall that $J_{P,2}^\star$ is concave, per Lemma \ref{lem: cvx_obj}. However, $J_{P,2}^\star - \kappa \widetilde{H}$ may not be concave. As a result, gaining more information may not always be advantageous for \PPP{}, given the cost of channel construction. Furthermore, $J_{P,2}^\star - \kappa \widetilde{H}$ is no longer piecewise affine--a property that is instrumental in allowing the computation of the bi-conjugate of $J_{P,2}^\star$ via the linear program in \eqref{prog: g3_agent}. The curvature of $\widetilde{H}$ prevents a direct computation of this bi-conjugate, which we only explore via graphical convexification of $J_{P,2}^\star - \kappa \widetilde{H}$ for specific examples.



Consider again our example where costs of \PPP{} and \AAA{} are given by \eqref{eq:ex.costs1} and an information acquisition cost with $\kappa=2$. Figure \ref{fig: g4_principal} illustrates the cost of \PPP{}, $J_{P,2}^\star - \kappa \widetilde{H}$ and its bi-conjugate $J_{P,4}^\star$ as a function of the posterior distribution induced by the information channel. \PPP{} chooses a channel such that the prior $\mu_0(\theta_1)=0.75$ is split into $\mu_s^+(\theta_1) = 0.87$ and $\mu_s^-(\theta_1) =0.01$. This leads to a cost of 1.88 for \PPP{} and 1.75 for \AAA{}. Compare this to $\mathfrak{G}_2$, where \PPP{} incurs a cost of 2, while \AAA{} incurs a cost of  1.75. Thus, \PPP{} benefits from constructing an information channel. However, her cost is higher in $\mathfrak{G}_4$ than in $\mathfrak{G}_3$, where she is able to achieve a cost of 1. Thus, with a  channel construction cost of $\kappa = 2$, \PPP{} would prefer the design that allows \AAA{} to report about the state rather than invest in constructing a channel.
We address games with a different cost structure in the next section.

\begin{figure}[h]
    \centering
    \includegraphics[scale=0.3]{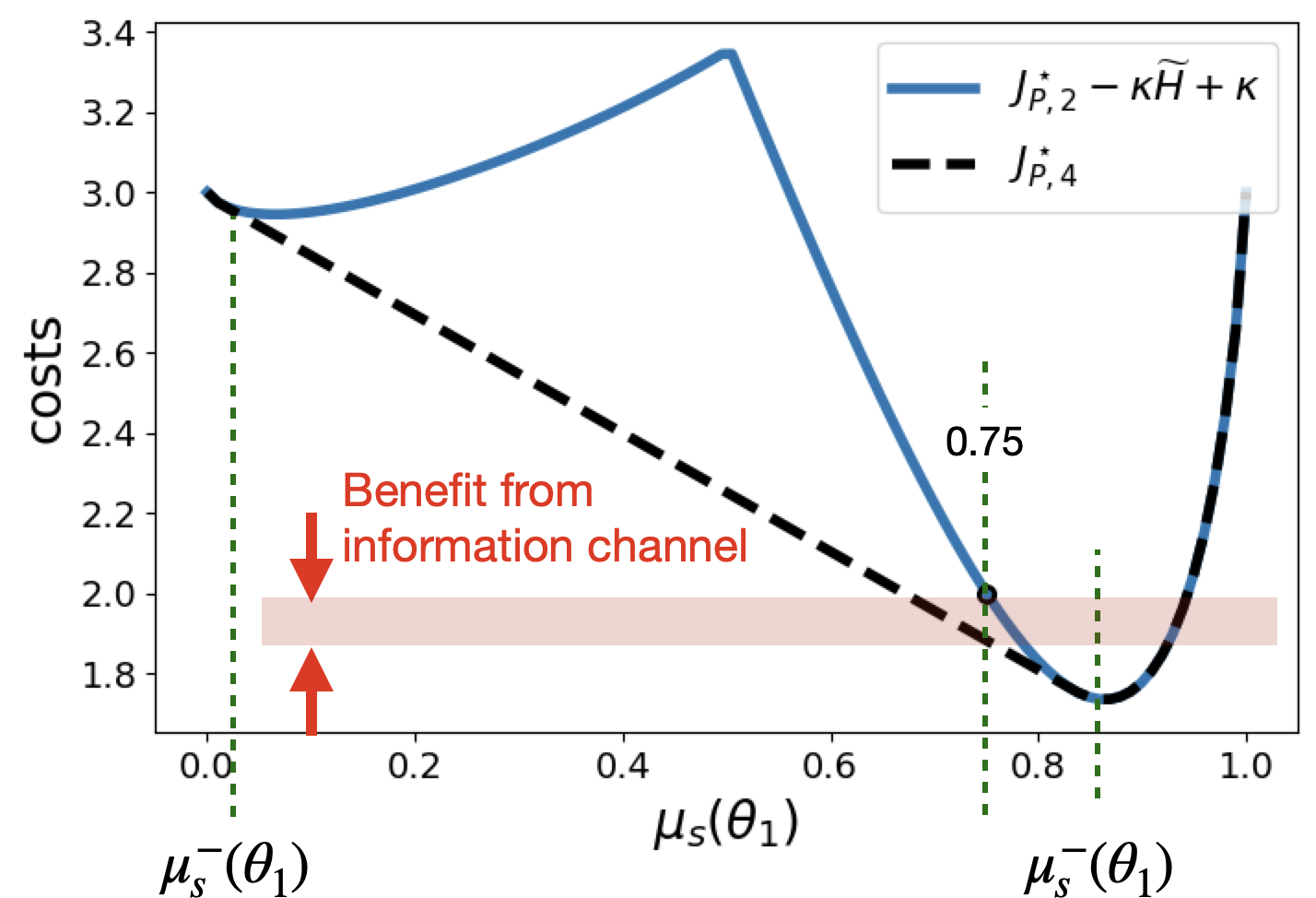}
    \caption{Convexification of \PPP{}'s cost considering information acquisition in game $\mathfrak{G}_4$.}
    \label{fig: g4_principal}
\end{figure}

\vspace{-0.5cm}
\section{$\mathfrak{G}_1$--$\mathfrak{G}_4$ in Quadratic Gaussian Games}
\label{sec:QG}
We now turn to study an example incentive design with information asymmetry and acquisition in a game 
with quadratic costs and continuous action spaces and state of the world. Specifically, let $\mathcal{U},\mathcal{V},\Theta$ be $\mathbb{R}$ and $\mu_0$ be a Gaussian distribution with mean $z_0$ and variance $\sigma_0^2$, i.e., $\mu_0 = \mathcal{N}(z_0,\sigma_0^2)$. Rather than a general theory, we will study a specific example and illustrate some of the nuances in studying information acquisition in QG games. We leave the general setup with vector state and action spaces to  future work. To that end, consider costs for \PPP{} and \AAA{}  given by 
\begin{align}
\begin{split}
    c_P(u,v;\theta)&=(\theta-u-v)^2+2u^2+\beta v^2,
    \\
    c_A(u,v;\theta)&=(\theta-u-v)^2+v^2.
\end{split}\label{Eq: Quad_costs}
\end{align}
Here, $\beta \geq 0$ is a cost related parameter.

In this section, we only consider deterministic maps of $v, \theta$ as candidate incentive policies for \PPP{}. First, we address the solution of $\mathfrak{G}_1$ for this QG game. Solving for a feedback
Stackelberg equilibrium requires \PPP{} to solve a functional optimization problem, which can be challenging. The work in \cite{bacsar1984affine} provides an alternative indirect mechanism to compute a Stackelberg equilibrium. In particular, 
\PPP{} herself solves for action $u^t$ for \AAA{} and $v^t$ for \PPP{} that together optimizes the cost for \PPP{}. {Since \PPP{}'s cost is strictly convex and quadratic, utilizing} first-order conditions, this \PPP{}-optimal solution for our example game is given by
\begin{align}
    u^t = \frac{\beta\theta}{3\beta+2}, \quad v^t = \frac{2\theta}{3\beta+2}.
\end{align}
Next, define an \emph{affine} incentive policy of the form, 
\begin{align}\label{Eq: g1_gamma_lq}
    \gamma(v; \theta) = u^t + Q (v - v^t),
\end{align}
where \PPP{} chooses $Q$ in such a way that the unique optimal response of \AAA{} to this incentive policy is $v^t$. Again, first order condition yields the optimal $Q^\star= (1-\beta)/\beta$, which in turn yields
\begin{align}
\begin{aligned}
    J_{P,1}^\star(\mu_0) &=\frac{2\beta(z_0^2+\sigma_0^2)}{3\beta+2},\\
    J_{A,1}^\star(\mu_0) &=\frac{4(\beta^2+1)(z_0^2+\sigma_0^2)}{(3\beta+2)^2}.
\end{aligned}
\end{align}


Next, we study $\mathfrak{G}_2$ within the same setup, where \AAA{} alone observes $\theta$ that is sampled from a commonly known prior $\mu_0$. This game involves a non-classical information structure, where \AAA{} has more information about the state than \PPP{} does, whereas \AAA{} has access to \PPP{}'s action. As a result the indirect mechanism for computing an optimal incentive policy $\gamma$ does not apply. Consequently, we restrict ourselves to specific functional forms for $\gamma(v)$ and study the resulting equilibrium. {However, the feedback incentive policy from $\mathfrak{G}_1$ is linear in $v,\theta$ and the coefficient of $v$ is independent of $\theta$. Motivated by this structure,} we consider \textit{affine mean feedback} policies for \PPP{} where the unknown realization of $\theta$ is replaced by its prior mean $z_0$ in \eqref{Eq: g1_gamma_lq}, which yields
\begin{align}\label{eq: affine_feedback}
    \gamma(v) = \frac{\beta^{2} z_0 - \left(\beta - 1\right) \left(v \left(3 \beta + 2\right) - 2 z_0\right)}{\beta \left(3 \beta + 2\right)},
\end{align}
for which the optimal response of \AAA{} is given by
\begin{align}\label{eq: vstar}
    v^\star(\theta) = \frac{\beta (3 \beta  + 2 )\theta -  (\beta^{2}   + 2 \beta - 2 )z_0}{ (\beta^2 + 1)( 3 \beta + 2)}.
\end{align}
%
Expected equilibrium costs for \PPP{} and \AAA{} in $\mathfrak{G}_2$ become
\begin{align}
    J_{P,2}^\star(\mu_0) &= \frac{2\beta}{3\beta + 2} z_0^2 + \frac{ \beta^{4} + \beta^{3} + 2 \beta^{2} - 4 \beta + 2}{\beta^{4} + 2 \beta^{2} + 1} \sigma_0^{2},
    \\
    J_{A,2}^\star(\mu_0) &= \frac{4  \left(\beta^{2} + 1\right)}{9 \beta^{2} + 12 \beta + 4} z_0^{2}
    + \frac{\beta^{2} }{\beta^{2} + 1}\sigma_0^{2}.
    \label{eq:G2.A.cost.QG}
\end{align}
Elementary algebra yields 
\begin{align}
\begin{aligned}
    J_{P,2}^\star - J_{P,1}^\star
    &= \frac{\left(\beta^{5} + 5 \beta^{4} + 4 \beta^{3} - 8 \beta^{2} - 4 \beta + 4\right)}{3 \beta^{5} + 2 \beta^{4} + 6 \beta^{3} + 4 \beta^{2} + 3 \beta + 2} \sigma_0^2
    \\
    &\geq 0.
\end{aligned}
\end{align}
The above inequality follows from the fact that the factor as a function of $\beta$ attains its minimum over $\beta \geq 0$ at $\sqrt{3}-1$, where this function vanishes. This inequality is expected as \PPP{} must pay a price for not knowing $\theta$ exactly. Moreover, the difference disappears when the variance of $\theta$ is zero. {We remark that affine mean feedback policies may not constitute an optimal affine incentive policy from \PPP{}'s standpoint. One can consider a more general affine policy of the form $\gamma(v)=L_1v+L_2$ and optimize over parameters $L_1,L_2$. While this route does not yield closed-form solution, a numerical comparison of \PPP{}'s performance with such policies and that with our mean-feedback policies is left for a future endeavor.}

Restricting attention to affine mean feedback strategies for \PPP{}, we now consider $\mathfrak{G}_3$, where \PPP{} allows \AAA{} to disclose information before committing to an incentive policy. Precisely, \PPP{} observes a signal $s$ sampled according to a signaling scheme $\pi$ chosen by \AAA{}. Restrict attention to signaling via an additive Gaussian noise channel, i.e., $s = \theta + w$, where $w \sim \mathcal{N}( 0, \sigma_w^2)$. In essence, \PPP{} only chooses $\sigma_w$ to fix the signaling mechanism. 

Given this signaling scheme, the posterior belief upon observing $s$ is given by 
\begin{align}
    \mu_s(s, \pi) = \mathcal{N}\left(\underbrace{\frac{\sigma_w^2 z_0+\sigma_0^2s}{\sigma_0^2+\sigma_w^2}}_{:=z_s},\underbrace{\frac{\sigma_0^2\sigma_w^2}{\sigma_0^2+\sigma_w^2}}_{:=\sigma_s^2}\right).
    \label{eq:QG.posterior}
\end{align}
With this updated belief, \PPP{} implements an affine mean feedback policy of the form \eqref{eq: affine_feedback} with $z_0$ replaced by $z_s$ to which \AAA{} responds with $v^\star(\theta)$ in \eqref{eq: vstar} with again $z_0$ replaced by $z_s$. Substituting the obtained affine mean feedback policy $\gamma(v)$ and $v^\star(\theta)$ in $c_A$ then yields an expression quadratic in $\theta, z_s$. \AAA{} must compute its expectation over $\theta \sim \mu_0$ and $s \sim \pi(\cdot|\theta)$, where $z_s = \E[\theta\mid s]$. Utilizing the relation,
\begin{gather}
    \E[\theta z_s] = \E[[\theta z_s \mid s ]] = \E[z_s^2] = z_0^2 + \frac{\sigma_0^4}{\sigma_0^2 + \sigma_w^2}
\end{gather}
and $\E[\theta^2] = z_0^2 + \sigma_0^2$, 
the expected cost faced by \AAA{} in $\mathfrak{G}_3$ becomes
\begin{align}
\begin{aligned}
    J_{A,3}^\star(\mu_0, \pi)
    = \left( \frac{\sigma_0^4}{\sigma_0^2 + \sigma_w^2} \right) f(\beta) + \eta,
\end{aligned}
\label{eq:JA3.QG}
\end{align}
where $\eta$ does not depend on $\sigma_w$ and
\begin{align}
    f(\beta):=
    \frac{4  \left(\beta^{2} + 1\right)}{9 \beta^{2} + 12 \beta + 4}   -  \frac{\beta^{2}  }{\beta^{2} + 1}.
\end{align}

Thus, \AAA{} either chooses $\sigma_w$ to be zero or infinity, corresponding to fully revealing $\theta$ or not revealing anything about $\theta$, depending on the sign of $f(\beta)$, a continuous function of $\beta$. When $f(\beta) > 0$, set $\sigma_w = \infty$ and choose $\sigma_w = 0$, otherwise. Notice that $
\lim_{\beta \to 0}f(\beta) = 1$
and $
\lim_{\beta \to \infty}f(\beta) = -5/9$, implying that the choice of revelation depends on the value of $\beta$. Such a result is reminiscent of pure persuasion in scalar QG games studied in \cite{tamura2018bayesian}. Hence, $J_{P,3}^\star$ is either $J_{P,2}^\star$ or $J_{P,1}^\star$ depending on no revelation or full revelation, respectively. {Since $J_{P,1}^\star$ is the team optimal cost for \PPP{} in QG games, when full revelation is optimal for \AAA{}, affine mean feedback strategies indeed are optimal for \PPP{}. When \AAA{} chooses not to disclose any information under mean feedback policies, however, there is scope for improvement with optimal affine or possibly nonlinear incentive policies--a topic we are keen to study in future work.}

Finally, consider $\mathfrak{G}_4$ where \PPP{} can pay to construct a channel to acquire information about $\theta$. Specifically, let the channel $\pi$ be of additive white Gaussian type that yields an observation $s = \theta + w$, where $w \sim \mathcal{N}(0, \sigma_w^2)$. The posterior distribution, contingent on the observation $s$, is  $\mu_s(s, \pi)$, described in \eqref{eq:QG.posterior}. With the affine mean feedback policy in \eqref{eq: affine_feedback}, \PPP{} faces a cost of the form \eqref{eq:G2.A.cost.QG} with $(z_0, \sigma_0^2)$ replaced by $(z_s, \sigma_s^2)$. Call this cost $J_{P,2}^\star(\mu_s(s,\pi))$. \PPP{} must pay for the information channel she wishes to use. Similar to Section~\ref{Sec: MatrixGames}, consider the channel cost to be equal to the decrease in entropy from a reference normal distribution $\mu'_0$, which, for this case, we consider to be $\mathcal{N}(0,1)$. This implies, any additive Gaussian noise channel maps the reference distribution to $ \mathcal{N}\left(\frac{\sigma_w^2s}{1+\sigma_w^2},\frac{\sigma_w^2}{1+\sigma_w^2}\right)$, where $s \sim \theta + \mathcal{N}(0,\sigma_w^2)$. Consequently, \PPP{} obtains the cost,
\begin{align*}
    J_{P,4}^\star (\mu_0, \pi)
    &=\E[J_{P,2}^\star(\mu_s(s,\pi))] + \frac{\kappa}{2}\log(1+\sigma_w^{-2})
    \\
    &=\frac{2\beta}{3\beta + 2} \left( z_0^2+\frac{\sigma_0^4}{\sigma_0^2+\sigma_w^2} \right)
    + \frac{\kappa}{2}\log(1+\sigma_w^{-2})
    \\
    & \quad + \frac{ \beta^{4} + \beta^{3} + 2 \beta^{2} - 4 \beta + 2}{\beta^{4} + 2 \beta^{2} + 1} \left(\frac{\sigma_0^{2}\sigma_w^2}{\sigma_0^2+\sigma_w^2}\right).
\end{align*}
In general, this is challenging to optimize. We plot this cost as a function of $\sigma_w$ for various values of $\beta$ in Figure~\ref{fig:g4_principal1} that reveals cases where noisy monitoring is optimal. For this figure, we chose $z_0=1,\sigma_0=2$. Notice that for $\beta = 1$, $f(\beta) < 0$ in \eqref{eq:JA3.QG}, implying that when left to \AAA{}, the latter will choose $\sigma_w = 0$ (full revelation), implying that \PPP{} will attain the cost in $\mathfrak{G}_1$. However, when \PPP{} invests in its information channel, her optimal choice from Figure \ref{fig:g4_principal1} will yield a noisy estimate of $\theta$ leading to a higher cost in $\mathfrak{G}_4$ than in $\mathfrak{G}_3$. As a result, \PPP{} stands to gain by allowing \AAA{} to signal. With $\beta = 0.5$, however, $f(\beta) > 0$, and \AAA{} will choose to garble information about the state via $\sigma_w = \infty$. Figure \ref{fig:g4_principal1} implies the optimality of a noisy estimate of $\theta$ to be optimal in $\mathfrak{G}_4$, leading to a better equilibrium cost. Thus, in this case,  \PPP{} must construct a channel rather than relying on \AAA{}. {Finally, as cost of information acquisition $(\kappa)$ increases, \PPP{} chooses not to acquire information, per Figure \ref{fig:g4_principal1}.}

\begin{figure}[h]
    \centering
    \includegraphics[scale = 0.4]{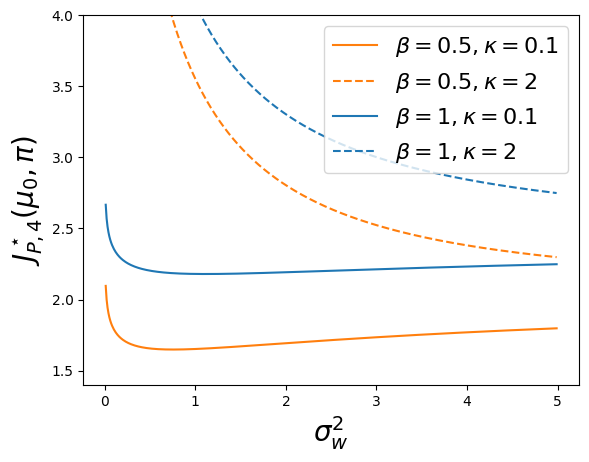}
    \caption{Variation of \PPP's cost in $\mathfrak{G}_4$ for various values of $\sigma_w^2$}
    \label{fig:g4_principal1}
\end{figure}

\section{Conclusions}
\label{sec:conc}

In this paper, we have formulated and studied an incentive design problem between a principal (\PPP{}) and an agent (\AAA{}), where we have endowed \AAA{} with informational advantage. Specifically, we have allowed \PPP{} to know only the prior distribution of a cost-relevant random variable (the state), while letting \AAA{} know the realization of the state. We have studied two types of information acquisition mechanisms. Namely, in one variant, \PPP{} allowed \AAA{} to reveal information about the state prior to the interaction, and in the other, invested in a channel to gain information about the state. We have studied these games within two setups--matrix games and quadratic Gaussian games. Our results reveal the trade-offs  between these two types of information acquisition techniques. {One direction for future work is to consider more general classes of policies for \PPP{} and study a broader class of QG games.} Another direction for future research concerns studying the same question in the context of repeated and dynamic games, where the state itself might evolve, \PPP{} can update beliefs over time, and \AAA{} can shape these beliefs.


\bibliographystyle{alpha}
\bibliography{references}

\newcommand{\etalchar}[1]{$^{#1}$}
\begin{thebibliography}{XFC{\etalchar{+}}16}

\bibitem[AC16]{alonso2016bayesian}
Ricardo Alonso and Odilon C{\^a}mara.
\newblock Bayesian persuasion with heterogeneous priors.
\newblock {\em Journal of Economic Theory}, 165:672--706, 2016.

\bibitem[Ake78]{akerlof1978market}
George~A Akerlof.
\newblock The market for “lemons”: Quality uncertainty and the market mechanism.
\newblock In {\em Uncertainty in Economics}, pages 235--251. Elsevier, 1978.

\bibitem[ALV23]{ali2023consumer}
S~Nageeb Ali, Greg Lewis, and Shoshana Vasserman.
\newblock Consumer control and privacy policies.
\newblock In {\em AEA Papers and Proceedings}, volume 113, pages 204--209. American Economic Association 2014 Broadway, Suite 305, Nashville, TN 37203, 2023.

\bibitem[AMS95]{aumann1995repeated}
Robert~J Aumann, Michael Maschler, and Richard~E Stearns.
\newblock {\em Repeated Games with Incomplete Information}.
\newblock MIT Press, 1995.

\bibitem[Ba{\c{s}}82]{8b51acce26a144b5949c35a7e184ada1}
Tamer Ba{\c{s}}ar.
\newblock General theory for {S}tackelberg games with partial state information.
\newblock {\em Large Scale Systems}, 3(1):47--56, 1982.

\bibitem[Ba{\c{s}}84]{bacsar1984affine}
Tamer Ba{\c{s}}ar.
\newblock Affine incentive schemes for stochastic systems with dynamic information.
\newblock {\em SIAM Journal on Control and Optimization}, 22(2):199--210, 1984.

\bibitem[Ba{\c{s}}24]{bacsar2024inducement}
Tamer Ba{\c{s}}ar.
\newblock Inducement of desired behavior via soft policies.
\newblock {\em International Game Theory Review}, 26(02):2440002, 2024.

\bibitem[BBM15]{bergemann2015limits}
Dirk Bergemann, Benjamin Brooks, and Stephen Morris.
\newblock The limits of price discrimination.
\newblock {\em American Economic Review}, 105(3):921--957, 2015.

\bibitem[BD04]{bolton2004contract}
Patrick Bolton and Mathias Dewatripont.
\newblock {\em Contract theory}.
\newblock MIT Press, 2004.

\bibitem[BHM22]{bergemann2022screening}
Dirk Bergemann, Tibor Heumann, and Stephen Morris.
\newblock Screening with persuasion.
\newblock {\em arXiv preprint arXiv:2212.03360}, 2022.

\bibitem[BO98]{bacsar1998dynamic}
Tamer Ba{\c{s}}ar and Geert~Jan Olsder.
\newblock {\em Dynamic Noncooperative Game Theory}.
\newblock SIAM, 1998.

\bibitem[CB82]{cansever1982minimum}
Derya~H Cansever and Tamer Ba{\c{s}}ar.
\newblock A minimum sensitivity approach to incentive design problems.
\newblock In {\em 1982 21st IEEE Conference on Decision and Control}, pages 158--163. IEEE, 1982.

\bibitem[CB85]{cansever1985optimum}
Derya~H Cansever and Tamer Ba{\c{s}}ar.
\newblock Optimum/near-optimum incentive policies for stochastic decision problems involving parametric uncertainty.
\newblock {\em Automatica}, 21(5):575--584, 1985.

\bibitem[CS06]{conitzer2006computing}
Vincent Conitzer and Tuomas Sandholm.
\newblock Computing the optimal strategy to commit to.
\newblock In {\em Proceedings of the 7th ACM conference on Electronic commerce}, pages 82--90, 2006.

\bibitem[DX16]{dughmi2016algorithmic}
Shaddin Dughmi and Haifeng Xu.
\newblock Algorithmic {B}ayesian {P}ersuasion.
\newblock In {\em Proceedings of the Forty-Eighth Annual ACM Symposium on Theory of Computing}, pages 412--425, 2016.

\bibitem[FT91]{fudenberg1991game}
Drew Fudenberg and Jean Tirole.
\newblock {\em Game Theory}.
\newblock MIT Press, 1991.

\bibitem[GK14]{gentzkow2014costly}
Matthew Gentzkow and Emir Kamenica.
\newblock Costly persuasion.
\newblock {\em American Economic Review}, 104(5):457--462, 2014.

\bibitem[Gro73]{groves1973incentives}
Theodore Groves.
\newblock Incentives in teams.
\newblock {\em Econometrica: Journal of the Econometric Society}, pages 617--631, 1973.

\bibitem[KG11]{kamenica2011bayesian}
Emir Kamenica and Matthew Gentzkow.
\newblock Bayesian persuasion.
\newblock {\em American Economic Review}, 101(6):2590--2615, 2011.

\bibitem[MM23]{matyskova2023bayesian}
Ludmila Matyskova and Alfonso Montes.
\newblock Bayesian persuasion with costly information acquisition.
\newblock {\em Journal of Economic Theory}, 211:105678, 2023.

\bibitem[MR78]{mussa1978monopoly}
Michael Mussa and Sherwin Rosen.
\newblock Monopoly and product quality.
\newblock {\em Journal of Economic Theory}, 18(2):301--317, 1978.

\bibitem[MR90]{milgrom1990rationalizability}
Paul Milgrom and John Roberts.
\newblock Rationalizability, learning, and equilibrium in games with strategic complementarities.
\newblock {\em Econometrica: Journal of the Econometric Society}, pages 1255--1277, 1990.

\bibitem[RS17]{roesler2017buyer}
Anne-Katrin Roesler and Bal{\'a}zs Szentes.
\newblock Buyer-optimal learning and monopoly pricing.
\newblock {\em American Economic Review}, 107(7):2072--2080, 2017.

\bibitem[SB20]{sayin2020persuasion}
Muhammed~O Sayin and Tamer Ba{\c{s}}ar.
\newblock Persuasion-based robust sensor design against attackers with unknown control objectives.
\newblock {\em IEEE Transactions on Automatic Control}, 66(10):4589--4603, 2020.

\bibitem[SCJ81]{salman1981incentive}
MA~Salman and Jose~B Cruz~Jr.
\newblock An incentive model of duopoly with government coordination.
\newblock {\em Automatica}, 17(6):821--829, 1981.

\bibitem[Spe78]{spence1978job}
Michael Spence.
\newblock Job market signaling.
\newblock In {\em Uncertainty in Economics}, pages 281--306. Elsevier, 1978.

\bibitem[Tam18]{tamura2018bayesian}
Wataru Tamura.
\newblock Bayesian persuasion with quadratic preferences.
\newblock {\em Available at SSRN 1987877}, 2018.

\bibitem[VBB25]{velicheti2025value}
Raj~Kiriti Velicheti, Melih Bastopcu, and Tamer Ba{\c{s}}ar.
\newblock Value of information in games with multiple strategic information providers.
\newblock {\em IEEE Transactions on Automatic Control}, 70(7):4532--4547, 2025.

\bibitem[XFC{\etalchar{+}}16]{xu2016signaling}
Haifeng Xu, Rupert Freeman, Vincent Conitzer, Shaddin Dughmi, and Milind Tambe.
\newblock Signaling in {B}ayesian {S}tackelberg {G}ames.
\newblock In {\em AAMAS}, pages 150--158, 2016.

\bibitem[ZB82]{zheng1982existence}
Ying-Ping Zheng and Tamer Ba{\c{s}}ar.
\newblock Existence and derivation of optimal affine incentive schemes for {S}tackelberg games with partial information: A geometric approach.
\newblock {\em International Journal of Control}, 35(6):997--1011, 1982.

\end{thebibliography}


\end{document}